\documentclass[11pt]{article}
\usepackage{graphics}
\usepackage{amssymb,amsfonts,amsthm}
\usepackage[english]{babel}
\usepackage[pdftex]{graphicx} 
\usepackage{geometry}
\usepackage[square,comma,numbers]{natbib}
\usepackage{etoolbox}
\usepackage{enumerate}
\usepackage{amsmath,latexsym}
\usepackage{amsmath}
\usepackage{amsfonts}
\usepackage{amssymb}
\usepackage{hyperref}
\usepackage{enumitem}
\usepackage{color}
\usepackage{mathrsfs}
\newtheorem{prop}{Proposition}
\newtheorem{defi}[prop]{Definition}
\newtheorem{cor}[prop]{Corollary}
\newtheorem{rem}[prop]{Remark}
\newtheorem{lem}[prop]{Lemma}

\usepackage{bbm}

\newcommand{\no}[1]{\mathop{\mathopen: {#1} \mathclose:}}

\def\be             {\begin{equation}}

	\def\ee             {\end{equation}}

\long\def\labl#1   {\label{#1}\ee}

\newcommand{\C}{\mathbb C}
\newcommand{\R}{\mathbb R}

\newcommand{\N}{\mathbb N}
\newcommand{\M}{\mathbb M}

\newcommand{\vol}{{\rm Vol}}

\newcommand{\loc}{{\rm loc}}
\newcommand{\reg}{{\rm reg}}
\newcommand{\mca}{\mu{\rm c}}

\DeclareMathOperator{\diag}{\rm diag}
\DeclareMathOperator{\WF}{\rm WF}
\DeclareMathOperator{\supp}{\rm supp}
\DeclareMathOperator{\Det}{\rm Det}

\newcommand\ria{\rightarrow}

\newcommand{\fA}{\mf A}

\newcommand{\eps}{\varepsilon}
\newcommand{\ph}{\varphi}

\newcommand{\tsf}[2]{{\textstyle{\frac{#1}{#2}}}}

\newcommand{\mf}[1]{\mathfrak{#1}}
\newcommand{\mc}[1]{\mathcal{#1}}

\newcommand{\cD}{\mc D}
\newcommand{\cE}{\mc E}
\newcommand{\cF}{\mc F}
\newcommand{\cT}{\mc T}
\newcommand{\cL}{\mc L}

\newcommand{\cS}{\mc S}
\newcommand{\Ci}{{\mc C}^\infty}

\newcommand{\T}{\cdot_{{}^\cT}}

\newcommand{\TF}{\cdot_{{}^{{\cT}_H}}}
\newcommand{\bTF}{\cdot_{{}^{\overline{{\cT}_H}}}}

\newcommand{\DD}{D_{\mathrm{D}}}

\newcommand{\sgn}{\mathrm{sgn}}

\newcommand{\bs}[1]{\boldsymbol{#1}}

\newcommand{\beqa}{\begin{eqnarray*}}
	\newcommand{\eeqa}{\end{eqnarray*}}

\newcommand{\beqan}{\begin{eqnarray}}
	\newcommand{\eeqan}{\end{eqnarray}}

\newcommand{\beq}{\begin{equation}}
	\newcommand{\eeq}{\end{equation}}

\newcommand{\benu}{\begin{enumerate}}
	\newcommand{\eenu}{\end{enumerate}}

\usepackage[colorinlistoftodos,textsize=footnotesize]{todonotes}

\begin{document}
\title{The Quantum Sine-Gordon model in perturbative AQFT\\
\Large Convergence of the S-matrix and the interacting current}
\author{Dorothea Bahns, Kasia Rejzner}
\date{}
%
%
%
%
\maketitle
\begin{abstract}
We study the Sine-Gordon model with Minkowski signature in the framework of perturbative algebraic quantum field theory. 
We calculate the vertex operator algebra braiding property. We prove that in the  finite regime of the model, the expectation value -- with respect to the vacuum or a Hadamard state --  of the  Epstein Glaser $S$-matrix and the interacting current or the field respectively,  both given as formal power series,  converge. 
\end{abstract}
\section{Introduction}
\label{intro}

Perturbative algebraic quantum field theory (pAQFT) is an approach to perturbation theory in quantum field theory that follows the paradigm of local quantum physics proposed by Haag and Kastler \cite{HK,Haag}. The important feature of this framework is that one separates the construction of the algebra of observables (local aspects of the theory) from the choice of a state (global features). This is of particular importance when generalizing the framework to quantum field theory on curved spacetime as advocated and pioneered in
\cite{BF0,BFV,HW01}. 
Since then, the pAQFT framework has been applied to a wide class of physical problems including quantization of a bosonic string \cite{BRZ,Z16} and effective quantum gravity \cite{BFRej13}. However, as pointed out in \cite{Summers},  up to now it has not been tested on an interacting model for which non-perturbative results exists. The present work means to bridge this gap and to establish convergence results in the massless Sine-Gordon model on the 2-dimensional Minokwski space in the model's ultraviolet-finite regime.  As it turns out, despite working with hyperbolic signature, we can still base our proofs of summability to the proof  established in \cite{Froe76} for a Euclidean version of the model. This way we not only test the robustness of the pAQFT framework, but also provide the first construction of the formal S-matrix in   the massless Sine-Gordon model on $\R^2$ (in the ultraviolet-finite regime) that is performed directly in the Lorentzian signature. We also construct the interacting currents and the interacting field. 

The pAQFT framework allows to construct the local algebras solely based on the fundamental solutions and solutions of the underlying linear (hyperbolic) partial differential equation. No Fock space is needed to calculate e.g. the $S$-matrix, which is given as a formal power series over a certain space of functionals. Moreover, it is not necessary to pass to a Wick rotated Euclidean version of the theory with an underlying elliptic PDE. Renormalization is formulated as a procedure of extending the time ordering map from regular functionals to local functionals, in the spirit of \cite{EG} (i.e. in particular in finite volume). The infrared cutoff is given by a compactly supported test function cutting off the interaction term, and it is usually kept throughout the calculations. The idea behind this is that local measurements do not depend on the particular choice of this test function. To calculate expectation values, a state is chosen, but it is not necessary to take this choice -- which on general manifolds is not canonical -- as the starting point of the construction. 

In the current paper, we work with the Sine-Gordon model in the regime of the coupling constant for which there is  no need for renormalization, so the time ordered products are constructed directly and, as distributions, they are shown to satisfy certain bounds. This makes it unnecessary to use the inductive procedure of Epstein and Glaser, but the formal $S$-matrix that we construct has the same physical interpretation as the object introduced in \cite{EG}.

We use the 2-point function of a Hadamard state for the massless scalar field in 2 dimensions to construct the star product of the free theory and, indirectly, to construct the time-ordered products. Such Hadamard states were given, for example, in \cite{Schubert,DM06}. The choice of the Hadamard state is not unique, but different choices lead to isomorphic algebraic structures.
In order to link our investigation directly to the calculations performed in the Euclidean setting \cite{Froe76,Col75}, we investigate also an alternative approach. We introduce an auxiliary finite mass $m$ and
study the $m\rightarrow 0$ limit of the expectation value of the S-matrix in the vacuum state for the massive theory. The resulting quantity can be interpreted as the expectation value of the S-matrix in a rather singular state, as clarified in sections 4 and 5.

The paper is organized as follows. We start by recalling the essential ideas and tools from pAQFT in the following (second) section. The third section is devoted to the Sine-Gordon model and more particularly, the vertex operators in our framework. In the fourth and fifth section we show summability of the $S$-matrix, the interacting current and the field itself.


\section{The framework of pAQFT}\label{sec:genpAQFT}

Let $\M_D$ be the $D$-dimensional Minkowski spacetime, i.e. $\R^D$ with the diagonal metric $\diag(1,-1,\dots,-1)$ and corresponding inner product denoted by $x \cdot x$. Starting point in the pAQFT construction of models is the classical configuration space $\cE$ that specifies the type of objects we want to describe. In general, it is the space of smooth sections of some vector bundle over $\M_D$. For the scalar field theory, we have $\cE\doteq\Ci(\M_D,\R$). We equip $\cE$ with its standard Frech{\'e}t topology.

Next, we consider the space of smooth functionals on $\cE$ (smoothness understood in the sense of \cite{Bas64,Ham,Mil,Neeb}). Among these, there are some important classes of functionals that are relevant for the construction of models in pAQFT. Firstly, we introduce the notion of local functionals.
\begin{defi}\label{localfunctionals}
	A functional $F\in\Ci(\cE,\C)$  is called local (an element of $\cF_{\loc}$) if for each $\ph_0\in\cE$ there exists an open neighbourhood $V$ of $\varphi_0$ in $\cE$ and $k\in\N$ such that for all $\ph\in V$ we have
	\begin{equation}
		F(\ph)=\int_{\M_D} \alpha(j^k_x(\ph))\ ,
	\end{equation}
	where $j^k_x(\ph)$ is the $k$-th jet prolongation of $\ph$ and $\alpha$ is a density-valued function on the jet bundle.
\end{defi}
The spacetime localisation of a functional is provided by the notion of \textit{spacetime support}
\begin{align}\label{support}
	\supp\, F\doteq\{ & x\in \M_D|\forall \text{ neighbourhoods }U\text{ of }x\ \exists \ph,\psi\in\cE, \supp\,\psi\subset U\,,
	\\ & \text{ such that }F(\ph+\psi)\not= F(\ph)\}\ .\nonumber
\end{align}

The notion of smoothness that we use implies that functional derivatives of a smooth functional 
$F\in\Ci(\cE,\C)$ can be seen as compactly supported distributions, i.e. 
$F^{(n)}(\ph)\in\Gamma'(\M_D^n,\R)^{\C}\equiv\cE'(\M_D^n)^{\C}$, 
where the superscript $\C$ indicates complexification. One can require a stronger condition, i.e. consider functionals whose derivatives are smooth. This motivates the following definition.
\begin{defi}
	A functional $F\in\Ci(\cE,\C)$ is called regular, i.e $F\in\cF_{\reg}$, if $F^{(n)}(\ph)\in\Gamma_c(\M_D^n,\R)^{\C}\equiv\cD(\M_D^n)^{\C}$ for all $n\in\N$, $\ph\in\cE$.
\end{defi}
More generally, one can impose different, less restrictive conditions on the regularity structure of functionals derivatives of functionals, seen as distributions. To describe the singularity structure, it is convenient to use H\"ormander's  wavefront (WF) set \cite{Hoer1}, a refined notion of the singular support of a distribution. It is a subset of the cotangent bundle whose projection onto the base is the singular support of the distribution (and whose covariables give the {\it high frequency cone}). It yields the following very simple sufficient criterion for the existence of products of distributions: if no two covariables (at the same base point) from the two respective wavefront sets can bee added to give 0, the product existst as a distribution.  Later on, we will see that the following class of functionals is a good choice for building models of pAQFT's; for other choices see \cite{BrDa,Yoann1,Yoann2}.
\begin{defi}\label{MicrocausalFunctioanls}
	A functional $F\in\Ci(\cE,\R)$ is called microcausal if it is compactly supported and satisfies
	\be\label{mlsc}
	\WF(F^{(n)}(\ph))\subset \Xi_n,\quad\forall n\in\N,\ \forall\ph\in\cE\,,
	\ee
	where $\Xi_n$ is an open cone defined as 
	\be\label{cone}
	\Xi_n\doteq T^*M^n\setminus\{(x_1,\dots,x_n;k_1,\dots,k_n)| (k_1,\dots,k_n)\in (\overline{V}_+^n \cup \overline{V}_-^n)_{(x_1,\dots,x_n)}\}\,,
	\ee
	where $(\overline{V}_{\pm})_x$ is the closed future/past lightcone understood as a conic subset of
	$T^*_x\M_D$.
\end{defi}
The construction of models in pAQFT starts with the free theory with the equation of motion of the form
\[
P\ph=0\,,
\]
where $P$ is a normally hyperbolic operator. For such operators there exist unique retarded (forward) and advanced (backward) fundamental solutions $\Delta^{\rm R}$, $\Delta^{\rm A}$ respectively. Their difference $\Delta=\Delta^{\rm R}-\Delta^{\rm A}$ is called the commutator function (or the causal propagator). As a distribution, $\Delta$ has WF set of the form
\[
\WF \Delta = \{(x,k;x',-k') \in \dot{T}^*\M^2_2 | (x,k) \sim (x',k')\}\,,
\]
where the equivalence relation $\sim$ means that there exists a null geodesic strip such that both $(x,k)$ and $(x',k')$ belong to it. One can then split \cite{Rad} $\Delta$ as a sum of two distributions
\be\label{dec}
\frac{i}{2}\Delta=W-H
\ee
in such a way that $H$ is symmetric and the WF set of $W$ (interpreted physically as the 2-point function) is 
\[
\WF W = \{(x,k;x',-k') \in \dot{T}^*\M^2_2 | (x,k) \sim (x',k'), k\in (\overline{V}_+)_x\}\,.
\]
The latter condition allows us to introduce the following non-commutative product  
for microcausal functionals 
$F,G\in \cF_{\mca}$,
\be\label{starprod}
(F\star_H G)(\ph)\doteq\sum\limits_{n=0}^\infty \frac{\hbar^n}{n!}\left<F^{(n)}(\ph),W^{\otimes n}G^{(n)}(\ph)\right>\,,\quad W=\tfrac{i}{2}\Delta+H\,.
\ee 
which gives a star product in the sense of formal power series on $\cF_{\mca}[[\hbar]]$.
The decomposition \eqref{dec} is in general not unique and different choices of $W$ are labeled by its different symmetric parts $H$. The difference $H-H'$ between two such choices is a smooth function, so the two star products $\star_H$ and $\star_{H'}$ are related by an equivalence in the sense of formal power series (``gauge transformation'')
$\alpha_H:\cF_{\mca}[[\hbar]]\rightarrow\cF_{\mca}[[\hbar]]$ 
given by 
\[
\alpha_{H-H'}\doteq e^{\frac{\hbar}{2}\cD_{H-H'}}\,,
\]
where, in terms of formal integral kernels (i.e. symbolic notation that makes the calculations easier),
\[
\cD_{H-H'}\doteq \left\langle H-H',\frac{\delta^2}{\delta\ph^2}\right\rangle=\int (H(x,y)-H'(x,y)) \frac{\delta^2}{\delta \varphi(x) \delta \varphi(y)} dxdy\]
and we have
\beq\label{eq:alphaHH}
F \star_{ H'} G = \alpha^{-1}_{H-H'}\left(\alpha_{ H-H'}(F)\star_{ H}\alpha_{ H-H'}(G)\right)\,.
\eeq

Note that the algebraic structure presented here looks the same for arbitrary dimension $D$ of $\M_D$, but the concrete form of $W$, $H$ and $\Delta$ would be different.

Let us now discuss how the star product $\star_H$ allows one to formulate the algebraic version of Wick's theorem and define Wick ordered quantities without using a concrete Hilbert space representation. We follow the construction introduced in \cite{DF,BDF}. Given a regular functional $F\in\cF_{\reg}$,  the formal power series formula for the Wick ordered expression  is 
\be\label{eq:Wick}
\no{F} \ =  e^{-\frac \hbar 2 \cD_H} F  = \alpha_H^{-1}F \in \mc F_{\reg}\,[[\hbar]]
\,,\quad \text { where } \quad \alpha_H\dot= e^{\frac \hbar 2 \cD_H}\,.
\ee
In order to extend this prescription to more general  (in particular non-linear local) functionals, one uses a limiting procedure. Consider  $\cF_{\mca}$ equipped with a topology $\tau_{\mathrm{Hoe}}$, which is a variant of the H\"ormander topology, see~\cite{BDF,BrDaHe,BrDa,Yoann1} for possible definitions of this topology and section 4.4.2 of \cite{Book} for a review.

Define ${\mathfrak A}_{\mca}[[\hbar]]$
as follows: take the completion of $\mc F_{\reg}$ (in the initial topology w.r.t. $\alpha_H$) and then select the subspace of all those $A \in  \overline{\mc F_{\reg}}[[\hbar]]$, $A=\lim A_n$ for some sequence in $\mc F_{\reg}[[\hbar]]$, which satisfy
\[
\alpha_H(A)=\lim_n \alpha_H(A_n)\in\mc F_{\mca} [[\hbar]]\,.
\]
 Now approximate $F\in\cF_{\mca}$ by regular functionals $F=\lim F_n$ where $F_n \in \mc F_{\reg}$ and define the corresponding normally ordered quantity as
\[
{\no{F}}=\lim \alpha_H^{-1}F_n \ \in {\mathfrak A}_{\mca}[[\hbar]]\,.
\]
We denote by $\fA_{\loc}[[\hbar]]$ the subspace of $\fA_{\mca}[[\hbar]]$ consisting of elements that arise as $\no{F}$, where $F\in\cF_{\loc}[[\hbar]]$.

On regular functionals we can introduce a star product that is independent of $H$:
\be\label{starprod2}
(F\star G)(\ph)\doteq\sum\limits_{n=0}^\infty \frac{\hbar^n}{n!}\left<F^{(n)}(\ph),\left(\tfrac{i}{2}\Delta\right)^{\otimes n}G^{(n)}(\ph)\right>\,,\quad F,G\in\cF_{\reg}\,.
\ee 
We then have
\beq\label{eq:AlphStarH}
\alpha_H\left(\no{F}  \star \no{G}\right) = F \star_H G \qquad (\mathrm{or}\quad\no{F}  \star \no{G} =\no{F \star_H G})
\eeq
where 
\[
F \star G =  \mu\circ e^{\frac {i\hbar } 2 D_\Delta }F \otimes G \in \mc F_{\reg}[[\hbar]]
\]
and
\[
F \star_H G = \mu\circ e^{\hbar D_W}F \otimes G 
\in \mc F_{\reg}[[\hbar]]
\]
with
\[
D_K(F\otimes G)= \left\langle K, \frac {\delta F }{\delta \varphi} \otimes \frac {\delta G}{\delta \varphi} \right\rangle = \int K(x,y) \frac{\delta F }{\delta \varphi(x)}\ \frac {\delta G} {\delta \varphi(y)} dxdy \qquad \text{ for kernel } K
\]
and where $\mu$ denotes the pullback along $\varphi \mapsto \varphi \otimes \varphi$.

This works also after the limiting procedure since
\[
\begin{array}{ccc}
(\mc F_{\reg}[[\hbar]], \star) & \stackrel{\alpha_H}{\longrightarrow} & (\mc F_{\reg}[[\hbar]], \star_H) \\
\downarrow && \downarrow\\
({\mathfrak A}_{\mca}[[\hbar]], \star) & \stackrel{\alpha_H}{\longrightarrow} & (\mc F_{\mca}[[\hbar]], \star_H) 
\end{array}
\]
where the down arrows are embeddings hence injective. The relation between $\star_H$ and $\star$ encodes the combinatorics of the Wick theorem. While the star product $\star$ is the ``standard'' product of the quantum theory, one can trade such products of Wick ordered quantities $\no{F}$ and $\no{G}$ for $\star_H$-products of ordinary functionals $F$ and $G$, using formula \eqref{eq:AlphStarH}. This is a big advantage of the pAQFT framework, since it allows us to write down concrete expressions and discuss convergence of formal power series without going to a Hilbert  space (or Krein  space as in~\cite{pier}) representation. The only input is the 2-point function $W$.

 Generally, a Gaussian state on $\mathfrak{A}_{\mca}[[\hbar]]$  with covariance $H$ is defined by evaluation of $\alpha_{H}(A)$ in a configuration $\varphi \in \cE$,
\beq\label{eq:vev}
\omega_{\varphi,H}(A)\doteq \alpha_{H}(A)(\varphi) \qquad \text{ for all } A \in {\mathfrak{A}}_{\mca}[[\hbar]]\,.
\eeq
The choice $\ph=0$ in the above is distinguished by the fact that $\omega_{0,H}$ is then exactly the expectation value in the state whose 2-point function is given by $W=\frac{i}{2}\Delta+H$ (see for example the discussion around formula (67) in \cite{LesHauches}, where the more complicated case of curved spacetimes is treated). As explained in section \ref{sec:genpAQFT}, instead of working in $(\mathfrak{A}_{\mca}[[\hbar]],\star)$ we can work in $(\cF_{\mca}[[\hbar]],\star_H)$ and motivated by the discussion above, we introduce the notation 
\[
F(\ph)=:\langle F\rangle_{\ph}\,,
\] 
for a functional $F\in\cF_{\mca}[[\hbar]]$.

The motivation behind the pAQFT approach is to make precise the Dyson formula for the scattering matrix and interacting fields. Consider scalar field on $D$-dimensional Minkowski spacetime $\M_D$. Recall that, heuristically, the Dyson formula for the interacting  time evolution operator $U_I(t,s)$ is
\[U_I(t,s)=1+\sum_{n=1}^\infty\frac{i^n\lambda^n}{n!}\int_ {([s,t]\times\mathbb R^{D-1})^n}T(\no{\cL_I(x_1)}\dots \no{\cL_I(x_n)})d^{D}x_1\dots d^Dx_n\,,
\]
where $\lambda$ is the coupling constant,  $T$ denotes time-ordering and the interaction Lagrangian $\no{\cL_I}$ is an operator-valued ``function''. This formula suffers from both UV and IR divergences. A way to give it mathematical meaning is to use the framework of Epstein and Glaser \cite{EG}. Here, the IR problem is solved by systematically treating the interaction Lagrangian as an operator valued distribution, and evaluating the $n$-fold time-ordered product $T(\no{\cL_I}\otimes \dots \otimes \no{\cL_I})$ in $g^{\otimes n}$, where $g$ is a compactly supported test function on $\M_D$. The UV divergences are controlled after carefully defining the time-ordered products.

Let us now recall this construction in the framework of pAQFT, see e.g. \cite{Book}. To avoid UV problems for the moment, we consider for now only regular functionals. Let  $F,G\in\cF_{\reg}$, then the time-ordered product of $F$ and $G$ is defined as 
\[
F\T G\doteq \mu\circ e^{i\hbar \DD}(F\otimes G)=\sum\limits_{n=0}^\infty \frac{\hbar^n}{n!}\left<F^{(n)},\left(i\Delta^{\mathrm{D}}\right)^{\otimes n}G^{(n)}\right>\,,\nonumber
\]
where $\Delta^{\mathrm{D}}$ is the Dirac propagator defined as
\beq\label{eq:Dirac}
\Delta^{\mathrm{D}}=\frac{1}{2}(\Delta^{\mathrm{R}}+\Delta^{\mathrm{A}})\,.
\eeq
More generally, we define  $n$-fold time-ordered products via a map $\cT_n : \cF_{\reg}^{\otimes n}[[\hbar]] \ria \cF_{\reg}[[\hbar]]$, given by
\[
\cT_n(F_1\otimes\dots \otimes F_n)\doteq F_1\T\dots\T F_n\,,\quad F_1,\dots,F_n\in\cF_{\reg}[[\hbar]]\,.
\]
Now consider an interaction $\no{V}$ with $V\in\cF_{\reg}$, where normal ordering given by the power series given in formula \eqref{eq:Wick} as $\alpha_H^{-1}V$.

Then the formal S-matrix is
\be\label{eq:S_exp}
\cS(\lambda \no{V})\doteq e_{\cT}^{i\lambda V/\hbar}=\sum_{n=0}^\infty \frac{1}{n!}(\tfrac{i\lambda}{\hbar})^n \underbrace{\no{V}\T\dots\T \no{V}}_n\,.
\ee

These formulae are well-defined for regular functionals, but usually, physically relevant interaction terms are local and non-linear, hence not regular. Therefore, one needs to extend $\cS$ from a map on $\cF_{\reg}[[\hbar]]$ to a map on $\mathfrak{A}_{\loc}[[\hbar]]$.
To this end, one first sets
\beq\label{eq:TnH}
\cT^H_n\doteq\alpha_H\circ\cT_n\circ (\alpha_H^{-1})^{\otimes n}
\eeq
and writes the S-matrix as
\be\label{eq:Smatrix}
\cS(\lambda \no{V})=\alpha^{-1}_H \sum_{n=0}^\infty \tfrac{1}{n!}(\tfrac{i\lambda}{\hbar})^n\ \cT^H_n(V^{\otimes n}) 
\doteq \alpha^{-1}_H \sum_{n=0}^\infty S_n(V)
\,.
\ee

It follows that extending $\cS$ to a map on $\mathfrak{A}_{\loc}[[\hbar]]$ (with values in $\mathfrak{A}_{\mca}[[\lambda]]((\hbar))$) is reduced to extending for any $n \in \N$, the time ordering $\cT^H_n$ to a map on $\cF_{\loc}[[\hbar]]$ (with values in $\cF_{\mca}[[\hbar]]$). This extension problem is called \textit{renormalization problem} and it is usually solved recursively using the Epstein Glaser procedure \cite{EG,BF0,BDF}. The combinatorial formula for $\cT_n^H$ is 
\beq\label{eq:TnH_comb}
\cT^H_n=e^{\hbar\sum_{i<j}\cD_{\mathrm{F}}^{ij}}\,,
\eeq
where the sum runs over all $1 \leq i < j \leq n$ and where $\cD_{\mathrm{F}}^{ij}\doteq \langle\Delta^{\mathrm{F}},\frac{\delta^2}{\delta\ph_i\delta\ph_j}\rangle$ and $\Delta^{\mathrm{F}}$ is the Feynman propagator defined by
\[
\Delta^{\mathrm{F}}:=\frac{i}{2}(\Delta^{\mathrm{R}}+\Delta^{\mathrm{A}})+H\,.
\]
Oberve that we use the term ``Feynman propagator'' to describe the bi-distribution obtained from the canonical, uniquely defined Dirac propagator by adding $H$. In this sense, the choice of the Feynman propagator is fixed by the choice of $H$, which in turn corresponds to the choice of the 2-point function $W$. This terminology, introduced in \cite{BDF}, is convenient in curved spacetimes, where, generically, there is no natural choice of $W$, hence no natural choice of the Feynman propagator. \medskip

In this paper, we treat the non-perturbative case, so we need to generalize the setting above to the situation where $\hbar$ is not a formal parameter, but a number. We proceed in a similar way. We equip the space of smooth functionals $\Ci(\cE,\C)$ with the topology $\tau$ of pointwise convergence of all the derivatives. The $n$-th functional derivative of $F\in\Ci(\cE,\C)$ at $\ph\in\cD$ is treated as an element of $\cE(\M_D^n)^{\C}$, equipped with the standard weak topology. We define $\fA$ similarly to $\fA_{\mca}[[\hbar]]$, by replacing $\cF_{\mca}[[\hbar]]$  with $(\Ci(\cE,\C),\tau)$.

\section{The Sine-Gordon model and vertex operators}
Starting point in the construction of the Sine-Gordon model is the free theory given by (minus) the massless Klein Gordon (i.e wave) operator $P=-\Box$ on two-dimensional Minkowski spacetime $\M_2=(\R^2,\eta)$, where $\eta=\diag(1,-1)$. 
Let $\Delta^{\rm R}$ and $\Delta^{\rm A}$ denote the retarded and advanced (forward and backward) fundamental solutions\footnote{To see that these are indeed fundamental solutions, observe that  $-\tsf 1 2 \theta(t-|\bs x|) = -\tsf 1 2 \theta(t-\bs{x})\theta(t+\bs{x})$ to calculate $\square (-\tsf 1 2 \theta(t-|\bs{x}|) )  = - \delta$ in the sense of distributions. Use $\partial_1 \theta(t-|\bs{x}|)=- \eps(\bs{x})\delta(t-|\bs{x}|)$ and $\partial_{1} \eps(\bs{x})= 2 \delta(\bs{x})$.} of $P$.

\begin{rem}[\bf Notation] 
	In flat spacetime we use the translation symmetry to express a translation invariant bi-distribution in terms of a distribution in one variable (the difference variable). By common  abuse of notation, we use the same symbols for the latter as for the corresponding bi-distributions, i.e.  $u(x,y)= u(x-y)$.
\end{rem}
Retarded and advanced fundamental solutions are given in terms of the following distributions in one variable:
\[
\Delta^{\rm R}(x)=-\tsf 1 2 \theta(t-|\bs x|) \qquad \Delta^{\rm A}(x)=-\tsf 1 2 \theta(-t-|\bs x|) \ , \text{ where } x=(t,\bs{x})  \in \M_2  \ .
\]
The 2-point function of the free massless scalar field $\varphi$ in 2 dimensions \cite{pier} coincides with the Hadamard parametrix \cite{Schubert}
\[
W(x)= -\frac 1 {4\pi} \ln\left(\frac{-x\cdot x + i \varepsilon t}{\Lambda^2}\right)
\]
where $\Lambda>0$ is the scale parameter. To work with dimensionless quantities, we choose the units so that we can set $\Lambda=1$. It is well known that $W$ is not a 2-point function of a Hadamard state, since, as a bi-distribution, it fails to be positive definite
 \cite{Wightman67}. Therefore, if one wants to study  representations of the abstract algebra from section~\ref{sec:genpAQFT}, one can apply the Krein space construction~\cite{MPS} or use the GNS representation of a Hadamard state. Note that the  2-point function of such a state differs from $W$ only by a smooth symmetric function on $\M_2$. Hence, the algebraic structure of the theory (e.g. the vertex operators algebra) can be studied independently of the choice of a Hadamard state.

We first write the 2-point function $W$ in terms of a symmetric ($H$) and antisymmetric ($\Delta$) contribution, $W=\frac i 2 \Delta + H$. The antisymmetric contribution is the causal propagator
\beq\label{eq:Causal}
\tsf i 2 \Delta = \tsf i 2 (\Delta^{\rm R} - \Delta^{\rm A}) 
\eeq
and a short calculation (see the appendix~\ref{app:calc}) then shows
\beq\label{eq:Hexplicit}
H=-\tsf {1}{4\pi} \left(\ln |t+\bs{x}|+\ln |t-\bs{x} |\right)=-\tsf {1}{4\pi} \left(\ln |x^2|\right)
\eeq
The wavefront sets of the distributions $\Delta$, $\Delta^{\rm A/\rm R}$
are such that H\"ormander's sufficient criterion for the existence of products of distributions does not apply, essentially because the wavefront set of the Heaviside function is that of the $\delta$-distribution. However, we will use that the characteristic function $\chi_{[0,\infty)}$ squares to itself and hence we set $\theta^2=\theta$ (as distributions). Similarly, we set the product of $\theta$ and $\theta\circ j$ where $j(x)=-x$, to 0, which is justified since the product of the locally integrable functions $\chi_{[0,\infty)}$ and $\chi_{(-\infty, 0]}$ is 0 in $L^1_{loc}$. 

Explicitly, we find for $n \in \N$,
\beqan
\Delta^n(x)&=&(- \tsf 1 2)^n(\theta(t-|\bs{x}|)-\theta(-t-|\bs{x}|))^n\nonumber\\
&=&(- \tsf 1 2)^n\theta(t-|\bs{x}|)+
(\tsf 1 2)^n\theta(-t-|\bs{x}|)\label{eq:ProdDelt}
\\
W^n(x)&=&\sum_{k=0}^n\left(n\atop k\right)\ \left(-\tsf i 4 \right)^{n-k}\left(\theta(t-| \bs{x}|)+
(-1)^{n-k}\theta(-t-|\bs{x}|)\right) \ H^k\label{eq:ProdW}
\,.
\eeqan
Note that the powers of the  2-point function $W$ are well-defined as usual, but contrary to the behaviour of massless fields in higher dimensions, also powers of the commutator function $\Delta$ are.

We will now study the exponential series of these distributions  in the following sense: 
Let $u \in \mc D^\prime(\R^k)$ be such that arbitrary powers $u^n$, $n \in \N$, are again distributions. Then one can investigate if the series 
\[
e^u(g):= \sum_n^\infty \tsf 1 {n!} u^n(g)
\]
converges (in $\R$) for any $g \in \mc D(\R^k)$. This is the case, for instance, if $u$ is smooth and polynomially bounded. In terms of formal integral kernels, we write
\[
e^u(g)= \sum_n^\infty \tsf 1 {n!} \int u(x)^n g(x) dx =: \int e^ {u(x)} g(x) dx 
\]
Observe that the last equality is simply short hand notation (not an interchanging of integration and taking sums). However, if $u$ is smooth and $e^{u(x)}$ converges
pointwise, then since $g$ is compactly supported, $e^{u(x)} g(x)$ converges uniformly (on $\supp g$) and in this case, the integral and the summation can indeed be interchanged.
In the same spirit we regard the identity
\[
(e^u  e^v )(g) =  \sum_n^\infty \tsf 1 {n!}  \sum_m^\infty \tsf 1 {m!} \int u(x)^n v(x)^m g(x) dx = e^{u+v}(g)
\]

This notation is now used to calculate  the commutation relation of two vertex operators $\no{V_a(f)}$. We first apply normal ordering to vertex operators,
\[
\no{V_{a}(x)}\equiv\no{V_{a}(t,\boldsymbol x)}\doteq \no{\exp(ia \Phi_{(t,\boldsymbol x)})}\,,
\]
where $a>0$,  $\Phi_{(t,\boldsymbol x)}$ is the evaluation functional, i.e. $\Phi_{(t,\boldsymbol x)}(\varphi)=\varphi(t,\boldsymbol x)$ for $\ph \in \cE$,  and where the Wick product is defined by~\eqref{eq:Wick}. These are the vertex operators. The smeared vertex operators are then defined by evaluating in a test function $g$
\beq\label{eqn:defVertex}
\no{V_a(g)}\doteq \no{\int \exp(ia \Phi_{x})g(x)dx} \in {\mathfrak A}_{\mu c}[[\hbar]]
\eeq
Here, we face the additional complication that the distributions 
depend on $\varphi\in \cE$. We will treat this dependence  pointwise as discussed at the end of section~\ref{sec:genpAQFT}.

In terms of formal integral kernels and using the formula for products of normally ordered functionals (\ref{eq:AlphStarH}) and evaluation in states~\eqref{eq:vev}, we find 
\beqan\label{eq:calcPropVert}
&&\alpha_H\left(\no{V_a(f)} \star \no{ V_{a'}(g)}\right) = \ \exp(ia\varphi)(f) \ \star_H   \exp(ia'\varphi)(g)\ = 
\nonumber
\\
&=&\sum_n \tsf{\hbar^n} {n!} \int \!\!...\!\int \prod_{j=1}^n W(x_j,y_j) \cdot
\frac {\delta^n \left(\sum_k\tsf {i^ka^k} {k!} \varphi^k \right)  }{\delta \varphi(x_1)... \delta \varphi(x_n)} (f)\cdot
\frac {\delta^n \left(\sum_k\tsf {i^ka'^k} {k!} \varphi^k \right)}{\delta \varphi(y_1)... \delta \varphi(y_n)}(g)\ d^{2n}xd^{2n}y
\nonumber
\\
&=&\sum_n \tsf{\hbar^n} {n!}  \int \!\!...\!\int \prod_{j=1}^n W(x_j,y_j) \cdot
\Big(\sum_{k\geq n}\tsf {i^ka^k} {k!} \tsf {k!}{(k-n)!}
\delta(x_1-x)\cdots \delta(x_n-x) \varphi^{k-n} f(x) \Big)  
\nonumber
\\
&&\qquad \qquad \qquad
\cdot \Big(\sum_{k\geq n}\tsf {i^ka'^k} {k!} \tsf {k!}{(k-n)!} \delta(y_1-y)\cdots \delta(y_n-y) \varphi^{k-n} g(y) \Big)    \ d^{2n}xd^{2n}y
\nonumber
\\
&=&\sum_n \tsf{\hbar^n} {n!} \int W(x,y)^n a^n a'^n i^{2n} \exp(ia\varphi(x))  f(x) \exp(ia'\varphi(y)) g(y)  \ .
\eeqan
In the spirit of the comments above we interpret the sum as the distribution $e^{-aa^\prime\hbar W(x,y)}$, so
\begin{align*}
\alpha_H(\no{V_a(f)} \star \no{ V_{a'}(g)}) &=
\int e^{-aa^\prime\hbar W(x,y)} f(x) g(y)\, V_a(x)  V_{a'}(y)  \ dxdy 
\nonumber
\\
&= \int  e^{-aa^\prime\hbar \frac i 2 \Delta (x,y) -aa^\prime\hbar H(x,y)} f(x) g(y)\,V_a(x)  V_{a'}(y) \ dxdy 
\nonumber
\end{align*}
Reversing the order of the vertex operators (i.e. interchanging $a$ and $a^\prime$ and $f$ and $g$, or the roles of $x$ and $y$) we find
\beqa
\alpha_H(\no{V_{a'}(g)} \star \no{ V_{a}(f)}) &=&
\int e^{-aa^\prime\hbar W(y,x)} f(x) g(y)\ 
V_{a'}(y) V_a(x) 
\ dxdy 
\\
&=& \int  e^{+aa^\prime\hbar (\frac i 2 \Delta (x,y) -aa^\prime\hbar H(x,y)} f(x) g(y)\    
V_{a'}(y) V_a(x)  \ dxdy 
\eeqa
and deduce that  in terms of formal integral kernels, the commutation relations are 
\[
\no{V_{a}(x)}\star\no{V_{a'}(y)}=  e^{-aa^\prime\hbar  i \Delta (x,y)}
\no{V_{a'}(y)}\star\no{V_{a}(x)}\, .
\]
Denote $x=(\bs{x},t)$, $y=(\bs{y},t')$. Using
\begin{multline*}
	e^{-aa^\prime\hbar  i \Delta (x,y)}=1-\theta((t-t')-|\bs{x}-\bs{y}|)-\theta(-(t-t')-|\bs{x}-\bs{y}|)
	\\+e^{aa'i\hbar/2}\ \theta((t-t')-|\bs{x}-\bs{y}|)
	+e^{-aa'i\hbar/2}\ \theta(-(t-t')-|\bs{x}-\bs{y}|)\,,
\end{multline*}
we see {\em a posteriori} that the exponential series is well defined in this case.
We get, in particular, for $\bs{x}=\bs{y}$ and $t>t'$,
\[
V_{a}(t,\bs{x})V_{a'}(t',\bs{x})=e^{aa'i\hbar/2}V_{a}(t',\bs{x})V_{a'}(t,\bs{x})\,,
\]
which is the well-known braiding property for vertex operators derived from e.g.~\cite[eqn 14.8.10]{kac}.

To further clarify the relation of our approach with the known literature on the Euclidean version of the Sine-Gordon model, we now calculate the vacuum expectation value of the product of 2 vertex operators (cp. \cite[Lemma 2.2]{Froe76}).  Let $W_v=W+v$ be a 2-point function of some Hadamard state for the appropriately chosen symmetric smooth function $v$. Let $H_v=H+v$. The new star product is equivalent to $\star_H$ and the equivalence is provided by the map $\alpha_v$. 

Let us now evaluate the product of  two vertex operators in the Gaussian state  with covariance $H_v$ as explained in \eqref{eq:vev} and calculate the vacuum expectation value (cp. \cite[Lemma 2.2]{Froe76}). For the latter, we find
\beqa
\omega_{0,H}(\no{V_{a}(f)}\star \no{V_{a'}(g)})&=&\left<\left(\int e^{ia \varphi(x)}f(x)dx\right)\star_H \left(\int e^{ia \varphi(y)}g(y)dy\right)\right>_0
\\&=& \int  e^{-aa^\prime\hbar W(x,y)} f(x) g(y)\   \ dxdy 
\\&=& \int f(x) g(y) ((t-t')^2-(\bs{x}- \bs{y})^2 + i (t-t')\eps)^{\hbar aa^\prime/4 \pi}\ dxdy
\eeqa
while the expectation value in the Hadamard state with covariance $H_v=H+v$ is
\[
\omega_{0,H_v}(\no{V_{a}(f)}\star \no{V_{a'}(g)})=\int e^{-aa'v(x,y)} \tilde{f}(x) \tilde{g}(y) ((t-t')^2-(\bs{x}- \bs{y})^2 + i (t-t')\eps)^{\hbar aa^\prime/4 \pi}\ dxdy\,,
\] 
where $\tilde{f}=e^{-a^2v}f$, $\tilde{g}=e^{-{a'}^2v}g$, so, in particular, the change of normal ordering of the vertex operators can be absorbed into the re-definition of test functions.  

\section{The pAQFT  $S$-matrix of the Sine-Gordon model}\label{sec:S_matrix}

Let us now apply the pAQFT framework to calculate the S-matrix $\cS$ for a potential of the form $V=\tfrac 1 2 (V_a + V_{-a})$. We will calculate the expectation value of $\cS(\lambda \no{V})$ in the coherent state $\omega_{\ph,{H_v}}$ where $H_v$ is the symmetric part of the 2-point function of some Hadamard state $W_v$. For now we assume that there is no need to renormalize -- which will a posteriori be justified by the calculations. First note that we can simplify the computation in the following ways:

{\rem\label{rem:expCalc} In our approach, $\alpha_{H_v}\circ\cS\circ\alpha^{-1}_{H_v}$ is a map on functionals and therefore, computing the expectation value in $\omega_{\ph,{H_v}}$ reduces to computing the evaluation of the functional $\alpha_{H_v}\circ\cS\circ\alpha^{-1}_{H_v}(\lambda V)$ at the configuration $\ph$. Also note that we can first construct $\alpha_{H}\circ\cS\circ\alpha^{-1}_{H}$ for $H$ given by \eqref{eq:Hexplicit} and then pass to  $\alpha_{H_v}\circ\cS\circ\alpha^{-1}_{H_v}$ using the intertwining map $\alpha_v$.}\bigskip

According to formula \eqref{eq:Smatrix}, the S-matrix is 
\beqan\label{eq:TnHSGV}
S(\lambda \no{V})&=&\alpha^{-1}_H\sum_n \tsf 1 {n!} 
(\tfrac{i\lambda}{\hbar})^n (\tsf 1 2)^n\ \cT^H_n((V_{a} + V_{-a})^{\otimes n}) 
\nonumber
\\ &=&
\alpha^{-1}_H\sum_n \underbrace{\tsf 1 {n!}  (\tfrac{i\lambda}{\hbar})^n (\tsf 1 2)^n \sum_{k=0}^n \left(n \atop k\right) \cT_n^H\Big(V_{a}^{\otimes k}\otimes V_{-a}^{\otimes  (n-k)}\Big)}_{\doteq  S_n(V)}
\,,
\eeqan
where the second identity follows from the fact that the time-ordered product is commutative.

We will now explicitly compute the S-matrix for the Sine-Gordon model and show that the resulting series converges in $\fA$ (as defined at the end of section \ref{sec:genpAQFT}). To calculate $\cT_n^H$, we first determine  the Feynman propagator of the Sine-Gordon model
\begin{equation}
\label{eq:DFexplicit}
\Delta^{\rm F}(t,\bs x)=\tsf i 2 (\Delta^{\rm R}+\Delta^{\rm A}) + H\,,
\end{equation}
which can also be written, using the well-known $\epsilon$-prescription as
\[
\Delta^{\rm F}(t,\bs x) = - \tsf 1 {4 \pi} \ln(x^2 - i\epsilon)\,,
\] 
cf. again the appendix~\ref{app:calc}. 

Using the explicit form 
of the Feynman propagator (\ref{eq:DFexplicit}), we find
\begin{multline}
\cT_n^H \left(V_{a_1} \otimes \dots \otimes V_{a_n}\right)(g^{\otimes n})= \int e^{i(a_1\varphi(x_1) + \dots + a_n\varphi(x_n))}\\
\times\prod_{1\leq i<j\leq n} 
\theta_{a_i a_j}(\tau_{ij}, \bs \zeta_{ij} ) \  \ 
|\tau_{ij}^2-\bs \zeta_{ij}^2|^{\hbar a_ia_j/4\pi}
g(x_1)\dots g(x_n) dx\label{eq:TnHexplicit}
\end{multline}
with $\tau_{ij}=t_i-t_j$ and $\bs \zeta_{ij}=\bs x_i - \bs x_j$ and with the abbreviation 
\[
\theta_{a_i a_j}(\tau, \bs \zeta ) = 
\theta(\tau+|\bs \zeta|) - \theta(\tau-|\bs \zeta|) + \left(\theta(\tau-|\bs \zeta|) + \theta(-\tau-|\bs \zeta|) \right)  \ e^{ia_i a_j\hbar /4}\,.
\]

To see this, observe that (with $\tau=t-t^\prime$, $\bs \zeta=\bs x - \bs x^\prime$) 
we have
\begin{multline*}
e^{- a a^\prime \hbar \Delta^{\rm F}(x,y)}
=
e^{ a a^\prime \hbar\left( \tsf i 4 (\theta(\tau-|\bs \zeta|) + \theta(-\tau-|\bs \zeta|) )\right)}
\ e^{ a a^\prime \hbar\left( \tsf 1 {4\pi} \ln |\tau^2-\bs \zeta^2|\right) } 
\\
=
\Big(
\underbrace{1-\theta(-\tau-|\bs \zeta|)}_{=\theta(\tau+|\bs \zeta|)} - \theta(\tau-|\bs \zeta|) + 
\left(\theta(\tau-|\bs \zeta|) + \theta(-\tau-|\bs \zeta|)\right) \ e^{\frac{a a^\prime  \hbar i} 4}
\Big) \ 
|\tau^2-\bs \zeta^2|^{\tsf {\hbar a a^\prime } {4\pi}}\,,
\end{multline*}
where the second equality follows from the idempotency of the Heaviside function, $e^{a\theta}=1+a\theta+\tfrac 1 2 a^2 \theta + \dots = 1 -\theta + \theta e^a $ for any real $a$.

Observe that for $a_i=-a_j$, the product $a_ia_j=-a^2$, and  the term
$|\tau_{ij}^2-\bs \zeta_{ij}^2|^{\hbar a_ia_j/4\pi}$ in \eqref{eq:TnHexplicit} is singular in $t_i-t_j=\pm (\bs x_i - \bs x_j)$. Using translation invariance we reduce this to a problem of a singularity along $\tau=\pm \bs \zeta$ (for a distribution defined on $\R^2$). This singularity is within the support of the Heaviside functions in \eqref{eq:TnHexplicit}. However, this is a homogeneous distribution which, as long as $a^2 \hbar/4\pi \not\in \N$,  has a unique extension according to~\cite[Thm~3.2.3]{Hoer1} and the prescription from~\cite{BF0}, see also~\cite{NST}. 
We will see below directly that these contributions are well-defined for $a^2 \hbar/4\pi <1$.

We now prove our main estimate, which gives a bound on the distributional product  occuring at $n$-th order perturbation theory as given in \eqref{eq:TnHexplicit}. This bound will enable us to prove that the S-matrix converges.

{\prop Let $\beta \doteq \hbar a^2/4\pi < 1$. Let $p > 1$ such that $\beta p< 1$. Let  $g \in C_c^\infty(\R^{2})$ be a function  cutting off the interaction $V$, and let $f=g^{\otimes n}$. Consider the  expectation value of the $n$-th order contribution to the  $S$-matrix   of Sine-Gordon theory in the state $\omega_{\varphi,H}$ with $H$ from \eqref{eq:Hexplicit}.  Choosing the support of $g$ small enough,  there is a constant $C=C(p,g)$ such that for all $n$,
	\[
	\Big| S_n(V)(f)(\varphi)\Big| 
	\leq \frac{\lfloor \frac n 2 \rfloor C^{{n}}\lambda^n}{(\lfloor \frac n 2 \rfloor{!})^{1-1/p}}\,,
	\]
	where $S_n(V)$ is given by \eqref{eq:TnHSGV}.
\label{prop:main}}

\begin{proof}
	
	According to remark \ref{rem:expCalc}, we indeed only need to evaluate the functional  $S_n(V)(g)$  at the configuration $\ph$.
	
	Now, we consider a contribution to \eqref{eq:TnHSGV} of the form $\cT_{n}^H\Big(V_{a}^{\otimes k}\otimes V_{-a}^{\otimes n-k}\Big)$ for $1 \leq k\leq n$ fixed.
	
	In the explicit form of the timeordering \eqref{eq:TnHexplicit}, we estimate the functions  $\theta_{a_i a_j}(\tau, \bs \zeta )$, which are given in terms of Heaviside functions and an oscillating factor, by 1. Without these functions, and neglecting the exponentials containing $\ph$,  the formal integral kernel 
	of $\cT_{n}^H\Big(V_{a}^{\otimes k}\otimes V_{-a}^{\otimes n-k}\Big)(\ph)$
	is
	\[
	\prod_{1\leq i<j\leq k} |\tau_{ij}^2-\bs \zeta_{ij}^2|^\beta
	\prod_{1\leq i\leq k, k< j\leq n} |\tau_{ij}^2-\bs \zeta_{ij}^2|^{-\beta}
	\prod_{k<i<j\leq n} |\tau_{ij}^2-\bs \zeta_{ij}^2|^\beta\,.
	\]
	with $\beta=\hbar a^2/4\pi>0$, and with the time variable differences $\tau_{ij}=t_i-t_j$ and the space variable differences $\zeta_{ij}=\bs{x}_i-\bs{x}_j$. 
	We rewrite this formula using 
	\[
	|\tau^2-\bs \zeta^2|=|\tau - \bs \zeta|\ |\tau + \bs \zeta|
	\]
	and consider the two factors separately, 
	\beq\label{eq:wpm}
	w^\pm_{n,k}(\underline{\tau},\underline{\bs \zeta})\doteq
	\prod_{1\leq i<j\leq k}|\tau_{ij}\pm\bs \zeta_{ij}|^{\beta}
	\prod_{1\leq i\leq k, k< j\leq n}  |\tau_{ij}\pm\bs \zeta_{ij}|^{-\beta}\prod_{k<i<j\leq n}  |\tau_{ij}\pm\bs \zeta_{ij}|^{\beta}
	\eeq
	with the underlined variables denoting the collection of all the respective difference variables.
	
	We now introduce new variables, $z, z^\prime \in\R^k$, $w, w^\prime \in \R^{n-k}$,
	\beqa
	z_i &=& t_i - {\bs x}_i \quad 1\leq i \leq k\\
	w_{i-k} &=& t_i - {\bs x}_i \quad k+1\leq i \leq n
	\\
	\\	z^\prime_i &=& t_i + {\bs x}_i \quad 1\leq i \leq k\\
	w^\prime_{i-k} &=& t_i + {\bs x}_i \quad k+1\leq i \leq n
	\eeqa
	Then we have 
	\[
	\tau_{ij}-\zeta_{ij}=\left\{\begin{array}{ll}
	z_i-z_j & 1\leq i<j \leq k\\
	z_i-w_{j-k} & i \leq k , k +1 \leq j\leq n\\
	w_{i-k}-w_{j-k} &  k +1 \leq i < j\leq n
	\end{array}\right.
	\]
	and likewise for the $+$-combination,
	\[
	\tau_{ij}+\zeta_{ij}=\left\{\begin{array}{ll}
	z^\prime_i-z^\prime_j & 1\leq i <j \leq k\\
	z^\prime_i-w^\prime_{j-k} & i \leq k , k +1 \leq j\leq n\\
	w^\prime_{i-k}-w^\prime_{j-k} &  k +1 \leq i < j\leq n
	\end{array}\right. \,.
	\]
	Therefore, we find
	\begin{equation}\label{w2kk}
	w_{n,k}^-(z,w)=\prod_{1\leq i<j\leq k} |z_i-z_j|^{\beta} \prod_{1\leq i < j \leq n-k}|w_i-w_j|^{\beta}
	\prod_{1\leq i\leq k, 1\leq j \leq n-k} |z_i-w_j|^{-\beta}
	\end{equation}
	and likewise for the $+$-combination, with the primed variables, $w_{n,k}^+(z^\prime,w^\prime)$.
	Note that for the negative powers we indeed get an unordered product (no relation $i<j$). This is a consequence of the fact that in the second factor in \eqref{eq:wpm}, the two indices are independent.
	Note also that for $k< n/2$ there are more $w$-variables than $z$-variables while for $k>n/2$ there are more $z$-variables. 
	
	In this notation, we find 
	\beqan
	|S_n(V)(f)(\ph)|&\leq & 
	\frac 1 {n!} \;
	C_0^n   \sum_{k=0}^n\left(n \atop k\right) \int  w^-_{n,k}(z,w) \,w_{n,k}^+(z^\prime, w^\prime)| f(z,w,z^\prime, w^\prime)| dz \dots dw^\prime 
\quad 	\label{eq:SnEstim}
	\eeqan
	where by abuse of notation, we use the same symbol  $f$ to denote $f$ in the new coordinates and where 
	$C_0= \frac{\lambda}{2\hbar}$.
	
	In this sum, we now consider a contribution with  $k\leq n/2$, i.e. one with less (or equally many) powers of  $V_{a}$ than powers of $V_{-a}$. Observe that this is  without loss of generality, since the contributions with $k>n/2$ can be treated in exactly the same way with the roles of $w$ and $z$ interchanged.

	We will see that in order to estimate the contributions with $k\leq n/2$, we need a slight generalization of the argument leading to ``The main estimate c)'' before Thm 3.4 in \cite{Froe76} -- an estimate which is still useful despite the fact that we started from a theory with hyperbolic signature.

	The starting point of our estimate is a generalization of the Cauchy determinant lemma, which we use to rewrite $w^-_{n,k}(z,w)$ as $|\Det\mathbf{D}|^\beta$, 
	where $\mathbf{D}$ is the $l\times l$-matrix ($l=n-k)$ with entries\footnote{We thank K.~Fredenhagen for pointing out this formula  to us.} 
	\[
	D_{ij}=\left\{\begin{array}{ccc}
	w_j^{i-1}&,&1\leq i\le l-k\ ,\\
	1/(z_{i-l+k}-w_j)\phantom{\int^z}&,&l-k<i\leq l\ .\end{array}\right.
	\]
	Note here that the number $l-k=n-2k$ counts how many more $w$ variables there are than $z$-variables. 
	
	We can proceed in exactly the same way for the $+$-combination $w_{n,k}^+(z^\prime, w^\prime)$.
	Since the plus and the minus case are exactly the same, we only discuss $w^-_{n,k}$ and  the integration over the unprimed variables. In notation, we suppress the unprimed variables, writing $F(z,w)$ for $f(z,w,z^\prime,w^\prime)$. Hence we consider 
	\[
	\int  w^-_{n,k}(z,w) |F(z,w)| dzdw = \int \left|\Det\mathbf{D}\right|^\beta \  |F(z,w)| dzdw
	\]
	We will now calculate an estimate on this integral's dependence on $n$,
	relying heavily on the  assumption that  $\beta\doteq\hbar a^2/4\pi < 1$. 
	
	Directly following the argument from \cite{Froe76}, we first 
	choose $p>1$ 
	with  $\beta p < 1$ and set $1<q\doteq \frac p {p-1} < \infty$. Viewing the compactly supported smooth function $F$ as an element of $L^q(\R^n)$, the H\"older inequality $\|FH\|_1= \int |FH|dx \leq \|F\|_q\|H\|_p$ for $\frac 1 q+ \frac 1 p=1$, yields
	\beqa
	\int \left|\Det\mathbf{D}\right|^\beta \  |F(z,w)| dzdw & \leq &  \ \|F\|_{L^q(\R^{n})} \  \left( \int_{K_g} \left|\Det\mathbf{D}\right|^{\beta p}dzdw \right)^{\frac 1 p} \,,
	\eeqa
	where $K_g$ is the compact set given by considering the support of $f (z,w,z^\prime, w^\prime)$ only with respect to the $z,w$ variables,
	\[
	K_g\doteq \Pi_{z,w}\supp f\doteq\{(z,w)\in\R^n| \ \text{there is } (z^\prime, w^\prime)\in \R^n \text{ s.t. } f(z,w,z^\prime, w^\prime)\neq 0 \}^{cl}  \,,
	\]	
	where $cl$ denotes closure and $\Pi$ stands for projection rather than product.
		 Observe that this is a compact set in $\R^n$ that really only depends on the test function $g$'s support in $\R^2$, since
	\[
	f(z,w,z^\prime,w^\prime)= \prod_{i=1}^k g\left(\tsf 1 2 (z_i+z_i^\prime), \tsf 1 2 (z_i-z_i^\prime)\right)\
	\prod_{j=1}^{n-k} g\left(\tsf 1 2 (w_j+w_j^\prime), \tsf 1 2 (w_j-w_j^\prime)\right) \, .
	\]
	Similarly, one finds that there is a constant, depending on the support of $g$ (and on $q$), such that 
	\beq\label{eq:est1}
	 \|F\|_{L^q(\R^{n})} \leq  (C_{g,q})^n\ .
	\eeq
	
	To compute the determinant, we use the Laplace expansion. In order to separate the contributions from the Vandermonde part of the matrix from those of the classic Cauchy matrix, we consider all matrices that are obtained by taking the first $l-k$ rows of $\mathbf D$ and choosing $l-k$ (of the $l$) columns of $\mathbf D$. Such matrices are called $\mathbf{D}_{\bf c}$ with $\bf c$ being the collection of labels denoting the choice of columns, ${\bf c}=(c_1,\dots,c_{l-k})$ with $1\leq c_1<\dots<c_{l-k}\leq l$,
	Let   $\mathbf{D}'_{\bf c}$ denote the $k\times k$-matrix that is complementary to $\mathbf{D}_{\bf c}$. Then up to a sign, the determinant of $\mathbf D$ is 
	\begin{equation}\label{Laplace}
	\sum_{\bf c}(-1)^{|\bf c|}\Det \mathbf{D}_{\bf c}\Det\mathbf{D}'_{\bf c} \,.
	\end{equation}
	where $|{\bf c}|=c_1+\dots+c_{l-k}$, and we find
	\[
	\int_{K_g} \left|\Det\mathbf{D}\right|^{\beta p}dzdw \leq 
	\sum_{\bf c} \int_{K_g} \left|\Det \mathbf{D}_{\bf c}\Det\mathbf{D}'_{\bf c} \right|^{\beta p}dzdw
	\]
	Here, we have used the triangle inequality and the fact that $\left(\sum |\dots|\right)^{\beta p}\leq \sum |\dots|^{\beta p}$ for $\beta p<1$.

	Note that each $\Det \mathbf{D}_{\bf c}$ is a Vandermonde determinant in the variables $w_i$, where $i\in{\bf c}$; and $\Det \mathbf{D}'_{\bf c}$ is a Cauchy determinant of a matrix with entries $1/(z_i-\tilde w_j)$, where $1\leq i,j \leq k$ and $(\tilde w_1, \dots, \tilde w_k)\doteq (w_i)_{i\in{\bf c^\prime}}$, where ${\bf c^\prime}\doteq \{1,\dots,l\}\setminus \bf c$ is the set of complementary column labels. 
	Therefore, we can factor each integral in the sum into a product of integrals over the Cauchy- and the Vandermonde determinant, respectively.	
	
	We first give the estimate on the Vandermonde determinants. Let $m\doteq  |{\bf c}|=l-k =n -2 k$, the number that counts how many more $w$-variables there are compared to $z$-variables. Then the Vandermonde Determinant is a homogeneous polynomial $P_V$ of degree $\tsf{1}{2}m(m-1)$ and we conclude that there are constants such that 
	\[
	\int_{K_{g, {\bf c}}} \left|P_V(u_1,\dots, u_m)\right|^{\beta p}du \leq 
	\vol(K_{g, {\bf c}}) \sup_{ K_{g, {\bf c}}} \left|P_V(u_1,\dots, u_m)\right|^{\beta p} \leq (C_{g,{\bf c}})^m \, 
	(\tilde C_{g,{\bf c}, p})^{\tsf{1}{2}m(m-1)}\,,
	\]
	where $K_{g, {\bf c}}$ denotes the projection of $K_g$ onto the variables $w_i$,  $i\in{\bf c}$, in the sense explained above,
	\[
K_{g, {\bf c}}\doteq \Pi_{{\bf c}}K_g\doteq \{ (w_i)_{i\in {\bf c}} \in \R^m | \   \text{there is } (z, (w_j)_{j\in {\bf c}^\prime})\in \R^{2k} \text{ s.t. } \chi_{K_g}(z,w)\neq 0 \}^{cl}\ .
\]	
	Observe that the constants depend on the choice of ${\bf c}$. By choosing the test function $g$ appropriately, we can make the constants arbitrarily small, so we get an estimate 
	\beq\label{eq:est2}
	\int_{K_{g, {\bf c}}} \left|P_V(u_1,\dots, u_m)\right|^{\beta p}du \leq 
	(C^{VdM}_{g,{\bf c}, p})^m
\eeq
where $VdM$ is short for Vandermonde.
	For the Cauchy determinants, we directly follow \cite[(3.15)]{Froe76}, and find an estimate
	\beq\label{eq:est3}
	\sum_{\pi \in S_k} \int_{K_{g,z}\times K_{g,{\bf c}^\prime}} \prod_{j=1}^k \frac 1 {|z_j-\tilde w_{\pi(j)}|^{\beta p}}dz d \tilde w
	\leq \sum_{\pi \in S_k} (C_{g,{\bf c},p})^k = k!\  (C_{g,{\bf c},p})^k
	\eeq
	where the integration is taken over all the variables $z$ and the variables $(\tilde w_1, \dots, \tilde w_k)\doteq (w_i)_{i\in {\bf c^\prime}}$ with ${\bf c^\prime}\doteq \{1,\dots,l\}\setminus \bf c$ and where $K_{g,z}$ and $K_{g,{\bf c}^\prime}$ are the corresponding parts of $K_g$ as explained above.
 To see the well-definedness of each contribution to the sum, recall that $\beta p <1$, hence the inverse powers are locally integrable.	Observe that the constant $C_{g,{\bf c},p}$ can be chosen independently of the permutation $\pi$ (by taking e.g. the maximum over permutations). 
	
	Repeating the argument for the $+$-combination $w_{n,k}^+(z^\prime, w^\prime)$ that depends on the primed variables, and putting together our three estimates \eqref{eq:est1}, \eqref{eq:est2} and \eqref{eq:est3}, we conclude: for $n$ fixed, we find constants independent of the choice of ${\bf c}$ (by taking the maxima, e.g. $\max_{{\bf c}}C_{g,{\bf c},p}$), such that  \eqref{eq:SnEstim} is estimated by
	\[
\frac {2}{n!}	\ C_0^n\ C_{g,q}^n\  \sum_{k=0}^{\lfloor \frac n 2 \rfloor} \binom{n}{k}\left(\binom{n-k}{n-2k}  k! \right)^{1/p}\ \underbrace{(C^{VdM}_g)^{\frac{n-2k}{p}} \; (C_{g,p})^{\frac k p}}_{\textstyle \leq (\underline{C}_{g,p})^{n-k}} 		\ ,
	\]
where the factor $\binom{n-k}{n-2k}$ arises from counting the number of terms in  the sum \eqref{Laplace} and the factor 2 accounts for the fact that we also have to consider the contributions with $k> n/2$. Observe that we have taken the $p$-th root from the H\"older estimate, also for the constants estimating the Vandermonde- and the Cauchy-Determinant. The constant $\underline{C}_{g,p}$ is the ($p$-th root of the) maximum of $C^{VdM}_g$ and $C_{g,p}$ and by an appropriate choice of the test function $g$ can be made smaller than 1. Hence, 
\[
(\underline{C}_{g,p})^{n-k} \leq (\underline{C}_{g,p})^{\lceil \frac n 2 \rceil} \qquad \text{ for any } 0\leq k \leq \lfloor \frac n 2 \rfloor \ .
\]
We now consider the sum of the factorials. Taking into account the $1/n!$ from the series expansion, this is
\[
\sum_{k=0}^{\lfloor \frac n 2 \rfloor} \frac 1 {n!}  \binom{n}{k}\left(\binom{n-k}{n-2k}  k! \right)^{1/p}\  =
\sum_{k=0}^{\lfloor \frac n 2 \rfloor}  \frac{1}{(n-k)!k!}\left(\frac{(n-k)!}{(n-2k)!} \right)^{1/p}\ 
\]
We can estimate this by
\[
\frac{\lfloor \frac n 2 \rfloor}{(\lfloor \frac n 2 \rfloor{!})^{1-1/p}}\,.
\]
We therefore find an estimate of the form
\[
...\leq 
	\frac{2 \lfloor \frac n 2 \rfloor}{(\lfloor \frac n 2 \rfloor!)^{1-1/p}}\  C_0^n C_{g,q}^n{(\underline{C}_{g,p})^{\lceil \frac n 2 \rceil}}
\]		
and deduce that there is a constant $C=C(g,p)$, such that 
	\[
	...\leq 	\frac{\lfloor \frac n 2 \rfloor C^{{n}}\lambda^n}{(\lfloor \frac n 2 \rfloor{!})^{1-1/p}}\,.
	\]
Observe that the dependence on $\lambda$ enters via $C_0=\lambda/2\hbar$.
	\flushleft\end{proof}

Observe that the constant in the estimate above is in general not smaller than 1. However, the factorial in the denominator dominates both this power and  the term linear in $n$  in the enumerator (use e.g. Stirling's formula) and we deduce the following result:

{\cor For an appropriate choice of the IR-cutoff function $g$ in the interaction $V$,  
the expectation value of the $S$-matrix in the coherent state $\omega_{\ph,H_v}$, for every $\ph\in\cE$, is summable for $\beta<1$, i.e. in the UV-finite regime.}
\begin{proof}
	From Theorem~\ref{prop:main} follows that the $S$ matrix is summable for  $\omega_{\ph,H}$ and $H$ differs from $H_v$ by a smooth function that can be absorbed into the redefinition of $F$. We can then choose the cutoff $f$ such that the support of $F$ is sufficiently small (as measured in the units of the scale parameter $\Lambda$) and use the same estimates as in the proof of Theorem~\ref{prop:main}.
\end{proof}

\bigskip
The above result guarantees the pointwise convergence of $\alpha_{H_v}\circ\cS\circ\alpha^{-1}_{H_v}(\lambda V)$ (i.e. the S-matrix twisted by $\alpha_{H_v}$) as a functional on $\cE$. The pointwise convergence of all the functional derivatives
follows from the fact that the $l$-th derivative of the vertex operator $V_{a}(f)$ in the direction of $\psi\in\cE$ is given by $(ia)^lV_{a}(f\psi)$, so one can use the same estimates as in Theorem~\ref{prop:main} to obtain bounds on the expressions of the form
\[
\left<\tsf{\delta^l}{\delta\ph^l}\cT_{n}^H\Big(V_{a}^{\otimes k}\otimes V_{-a}^{\otimes n-k}\Big),\psi^{\otimes l}\right>(\ph)
\]
and the convergence of the formal $S$-matrix
in the topology discussed at the end of section~\ref{sec:genpAQFT} follows.

\medskip 
Let us now discuss another possible choice of a state for  the construction of the S-matrix, which is closer to the Euclidean setting. As proposed e.g. in \cite{Wightman67}, one can introduce an auxiliary mass $m$, compute vacuum expectation value of the time-ordered products of the massive theory, and study the $m\rightarrow 0$ limit. We will adopt this strategy, also advocated by  \cite{Col75,Froe76}, in our consideration to make comparison with standard results from Euclidean theory. After constructing the time-ordered products with the Feynman propagator of the massive theory, we will show that the mass zero limit exists and that the mass zero time-ordered products satisfy certain bounds, which allow us to show the convergence of the formal $S$-matrix. Since the massless field in 2 dimensions does not possess a vacuum state, this limit is singular, leading in particular to the vanishing of all odd contributions to the $S$-matrix.
The final step, which we do not take here, would be to take the adiabatic limit of the  $S$-matrix by removing the cutoff function form the interaction term. 

Instead of directly calculating  \eqref{eq:TnHexplicit}, we therefore study a massive counterpart of auxiliary finite mass $m$ (later to be taken to zero again). In the present framework (of pAQFT), this is achieved by taking an appropriate massive Feynman propagator to define the time ordered products. A first attempt is to take the massive propagator that corresponds to $W_m$ that is the standard Wightman 2-point function,
\[
\Delta^{\rm F}_m(z)\doteq\frac 1  {2\pi}K_0(mz)\,, 
\]
where $z=\sqrt{-x^2+i\epsilon}$.  However, as the limit $m \ria 0$ of this propagator does not exist, we consider instead the propagator proposed in~\cite[eqn. (97)]{BDF} -- which comprises an additional part $h$, 
\[
h(z)\doteq\frac{1}{2 \pi}\, \ln (m/\mu) \, I_0(mz)
\]
with some real positive parameter $\mu$ and the modified Bessel function of the first kind of order 0.
The resulting propagator
\[
\widetilde{\Delta}^{\rm F}_m(x)=\Delta^{\rm F}_m(x)+h(z)=\frac{1}{2 \pi}\left(K_0(mz)+\ln (m/\mu) I_0(mz)\right)\,, \ \text{where } z=\sqrt{-x^2+i\epsilon}
\]
is well-defined for any mass $m \geq 0$ and indeed, the limit $m \ria 0$ yields our massless Feynman propagator up to an additive constant,
\[
\widetilde{\Delta}^{\rm F}_0=\Delta^{\rm F}-\frac{1}{2\pi}\ln(\mu/2)+\gamma\,,
\]
where $\gamma$ is the Euler constant.
Since $h$ is smooth, we can use it to define an equivalence relation between the two time ordered products w.r.t. $\Delta^{\rm F}_m$ and $\widetilde{\Delta}^{\rm F}_m$. Let $H_m=\Delta^{\rm F}_m-\frac i 2 (\Delta_m^{\rm R} + \Delta_m^{\rm A})$ with the corresponding massive advanced and retarded propagators. The S-matrix of the massive theory is
\[
\cS_m(\lambda \no{V})=\alpha^{-1}_{H_m} \sum_{n=0}^\infty \tfrac{1}{n!}(\tfrac{i\lambda}{\hbar})^n\ \cT^{H_m}_n(V^{\otimes n})\,.
\]
We take the vacuum expectation value of $\cS_m(\lambda \no{V})$  according to \eqref{eq:vev} in the state given by 
\[
\omega_{\varphi, H_m}(A)=\alpha_{H_m}(A)(\varphi)\, ,
\]
 To realize the modification of the Feynman propagator, we change the prescription of normal ordering and define new normal ordered operators as
\[
\no{A}\doteq \alpha_{H_m+h}^{-1}(A) \,, A\in\cF_{\mca}[[\hbar]].
\]
Note that (set $\mu=1$)
\beq\label{eq:NOmassive}
\alpha_h^{-1}(e^{ia \Phi_x})(\ph)=e^{\frac{\hbar}{4 \pi}\, a^2\, \ln (m) +ia\ph(x)}=m^{\frac{\hbar }{4 \pi}\, a^2 +ia\ph(x)}\,,
\eeq
where $\Phi_x$ is the evaluation functional (i.e. $\Phi_x(\ph)=\ph(x)$).

Hence we compute
\[
\omega_{\ph,H_m}(\cS_m(\lambda \no{V}))= \sum_{n=0}^\infty \tfrac{1}{n!}(\tfrac{i\lambda}{\hbar})^n\ \cT^{H_m}_n((\alpha_h^{-1}V)^{\otimes n})\,,
\]
and ultimately, we want to take the limit
\[
\lim_{m\rightarrow 0}\omega_{\ph,H_m}(\cS_m(\lambda \no{V}))\,.
\]
Using the formulae for the time ordering 
\eqref{eq:TnH} and
\eqref{eq:TnH_comb} and the explicit form of $\Delta^{\rm F}_m$, as well as the explicit form \eqref{eq:NOmassive} of the equivalence map $\alpha_h^{-1}$ acting on exponentials, we find
\beqa
&&\cT_n^{H_m} \left(\alpha_h^{-1}(V_{a_1}) \otimes \dots \otimes \alpha_h^{-1}(V_{a_n})\right)(g^{\otimes n})
\\ &=& 
\int
e^{i\sum_ia_i\varphi(x_i)}\, m^{\frac{\hbar}{4 \pi}\, \sum_i a_i^2}\, e^{- \frac \hbar{2 \pi}\sum_{i<j} a_i a_j K_0(m z_{ij})}f(x_1)\dots f(x_n) dx
\eeqa
where $dx\doteq dx_1\dots dx_n$, and $z_{ij}=\sqrt{-(x_i-x_j)^2+i\epsilon}$, and where the minus sign in the exponential involving the Feynman propagator $K_0$ is $i^2$ from  the functional derivative acting on the vertex operators.
Our auxiliary mass $m$ being arbitrarily small and all arguments $x_i-x_j$ being bounded by the support properties of our test functions $g$, we now use the expansion for $K_0$ for small argument, $K_0(mz)=-(\ln(mz/2)+\gamma)I_0(mz)$, 
\[
I_0(mz)=1+\textrm{ a power series in }mz\textrm{ starting with a quadratic term},
\] to rewrite the Feynman propagator $\Delta^{\rm F}_m$.
We then find, with $z_i$ and $z_{ij}$ as above,
\begin{multline*}
\cT_n^{H_m} \left(\alpha_h^{-1}(V_{a_1}) \otimes \dots \otimes \alpha_h^{-1}(V_{a_n})\right)(g^{\otimes n})
\\
=\int e^{i\sum_ia_i\varphi(x_i)} \,  m^{\frac{\hbar}{4 \pi}\, \sum_i a_i^2 }\qquad\qquad\qquad\qquad\qquad\qquad\\
\times e^{\frac \hbar {2\pi}\sum_{i<j} a_i a_j (\ln (m z_{ij}/2) + \gamma)I_0(mz_{ij})}g(x_1)\dots g(x_n) dx
\\
= \int e^{i\sum_ia_i\varphi(x_i)} \, m^{\frac{\hbar}{4 \pi}\, \left(\sum_i a_i^2 + 2 \sum_{i<j}a_ia_j I_0(mz_{ij})\right)}\qquad\quad\\
\times e^{ \frac \hbar {2\pi}\sum_{i<j} a_i a_j (\ln ( z_{ij}/2) + \gamma)I_0(mz_{ij})}g(x_1)\dots g(x_n) dx\,.
\end{multline*}
Since $I_0(mz) \ria 1$ in the limit $m \ria 0$, only those terms where $\sum_i a_i^2  + 2 \sum_{i<j}a_ia_j = (a_1 + \dots + a_n)^2=0$ survive. In particular, all the odd $n$ contributions vanish.

The remaining terms are then of the form 
\[
\cT_n^H \left(V_{a_1} \otimes \dots \otimes V_{a_n}\right)(f^{\otimes n})= \int e^{i(a_1\varphi(x_1) + \dots + a_n\varphi(x_n))} e^{- \hbar\sum_{i<j} a_i a_j\Delta_F(x_i,x_j)}g(x_1)\dots g(x_n) dx
\]
with $n$ even, $\sum_i a_i=0$ and with our Feynman propagator $\Delta^{\rm F}$ from \eqref{eq:DFexplicit}.
Hence, this is a special case of \eqref{eq:TnHexplicit}. Thus, we recover the following well-known result on the summability of the $S$-matrix in the vacuum state as a special case of our main estimate.

{\rem The estimate in Theorem~\ref{prop:main} is improved if we consider the expectation value in the singular state obtained by the limit $\lim_{m\rightarrow 0}\omega_{\ph,H_m}(\cS_m(\lambda \no{V}))$. In this case, the only non-vanishing contributions occur for even $n$ and $k=l=n/2$, and the matrix $\bf D$ is just the Cauchy matrix with no Vandermonde determinants present. Hence, exactly as in~\cite{Froe76}, no restriction on the support of the test function $g$ is necessary. We expect that a similar result can be obtained also by choosing an appropriate class of Hadamard states.}

\section{Interacting fields}
In the framework of pAQFT interacting fields are constructed with the use of the Bogoliubov formula \cite{BS,DFloop} (see also \cite{Book} for a review)
\[
F_I\doteq -i\hbar \frac{d}{d t}\cS(\lambda \no{V})^{\star -1}\star \cS(\lambda\no{V}+t\no{F})\Big|_{t=0}\,,
\]
where $F$ is a classical observables (a functional in $\cF_{\mca}$). This formula has to be understood as a formal power series in $\lambda$ 
\be\label{eq:FI}
F_I=\sum_{0}^n\frac{\lambda^n}{n!}R_n({\no{V}}^{\otimes n},\no{F})
\ee
and coefficients of this series are called retarded products and are given by
\beq\label{eq:Rn}
R_n({\no{V}}^{\otimes n},F)=\left(\tsf i \hbar\right)^n \sum_{k=0}^n \left( n \atop k \right) (-1)^k  \bar\cT_k({\no{V}}^{\otimes k})\star \cT_{n-k+1}({\no{V}}^{\otimes (n-k)}\otimes \no{F})\,,
\eeq
where $\bar\cT_k$ are the antichronological products defined as the coefficients in the expansion of the inverse (in the sense of formal power series) S-matrix, i.e.
\[
\cS(\lambda\no{V})^{\star-1}=\sum_{n=0}^{\infty}\tsf 1 {n!}\left(\tsf{-i\lambda}{\hbar}\right)^n\bar{\cT}_n({\no{V}}^{\otimes n})\,.
\]
In the standard Epstein-Glaser construction, the antichronological products are constructed inductively, together with the time ordered products. The prove of the induction step relies on the fact that one can perform the causal split of relevant distributions \cite{Scharf}, or find appropriate distributional extensions \cite{Ste71}. Here, however, we can use the ``naive'' split of the causal propagator, obtained by multiplying with the Heaviside theta distributions. It remains to prove that all the expressions obtained in this way are well defined. This is not entirely obvious (would in fact fail in higher dimensions), so we will write these expressions down explicitly.

Recall that the anti Feynman propagator is defined by
$\Delta^{\rm{AF}} = - i \Delta^{\rm D} + H = \overline{\Delta^{\rm F}}$,
\beqan\label{eq:DAFexplicit}
\Delta^{\rm{AF}}(t,\bs x)&=& 
\tsf i 4 (\theta(t-|\bs x|) + \theta(-t-|\bs x|) )
- \tsf 1 {4\pi} (\ln |t+\bs x| + \ln|t- \bs x|) 
\nonumber \\
&=&+ \tsf i 4 \theta(t^2-|\bs x|^2) 
- \tsf 1 {4\pi} \ln |t^2-\bs x^2|
\eeqan
We will use the following identity, which we prove in the appendix. Let $\Delta^{\rm D}$ denote the Dirac propagator \eqref{eq:Dirac} which defines the time ordering. Then  for $n \geq 0$, we have
	\beq\label{eq:DeltaDn}
	2^n\Delta^{\rm D}(x)^n = \theta(t) \Delta(x)^n + \theta(-t)\Delta(-x)^n
	\eeq
	where $\Delta$ denotes the causal propagator \eqref{eq:Causal}, which is a sum of two Heaviside functions.

It follows  -- 
as in  the calculation leading to \eqref{eq:calcPropVert} -- that the time ordered product of 2 vertex operators satisfies (in the sense of formal integral kernels)
\begin{multline*}
\alpha_H(\no{V_{a_1}}(x_1) \T \no{V_{a_2}}(x_2)) = e^{i a_1 \varphi( x_1) + i a_2 \varphi( x_2)}\ e^{-ia_1 a_2\Delta_D(x_1-x_2)-a_1a_2H(x_1,x_2)}
\\
\stackrel{\eqref{eq:Diracn}}{=}
e^{i a_1 \varphi( x_1) + i a_2 \varphi( x_2)-a_1a_2H(x_1,x_2)}\ 
\left(
\theta(\tau_{12})\ e^{-\frac i 2 a_1 a_2\Delta(x_1-x_2)}+\theta(-\tau_{12}))\ e^{-\frac i 2 a_1 a_2\Delta(x_2-x_1)}
\right)\\= \alpha_H\left(\no{V_{a_1}}(x_1) \star \no{V_{a_2}}(x_2)\theta(\tau_{12})+\no{V_{a_2}}(x_2) \star \no{V_{a_1}}(x_1)\theta(-\tau_{12})\right)\,.
\end{multline*}
Contrary to the typical situation in higher dimensional models \cite{Scharf}, the products of distributions in the above formula are well defined and no renormalization is needed. This generalizes to $k$-fold time-ordered products.
\begin{equation}\label{chron}
\alpha_H(\no{V_{a_1}}(x_1) \T \dots\T \no{V_{a_k}}(x_k))=
\sum_{\sigma\in S_k} V_{a_{\sigma(1)}}(x_{\sigma(1)})\star_H\dots\star_H V_{a_{\sigma(k)}}(x_{\sigma(k)})\prod_{i=1}^{n-1}\theta(\tau_{\sigma_{i}\sigma_{i+1}})\,. 
\end{equation}
where $S_k$ is the group of permutations of $k$ elements. 
Using the classic argument, we see that the antichronological product is the same, up to the order of the arguments in the Heaviside function.
Now, we apply again the identity \eqref{eq:DeltaDn} and see that
for $V=\frac 1 2 (V_a+V_{-a})$ (which is real),
\[
\alpha_H\circ\bar{\cT}_n(\no{V}(x_1)\otimes\dots\otimes \no{V}(x_n)) = \overline{
	V(x_1) \TF \dots \TF V(x_n)
}=V(x_1)\bTF\dots\bTF V(x_n)\,,
\]
where, in the second step, we used that $H$ is real and where $\bTF$ is defined
by the exponential formula 
\[
F\bTF G\doteq \mu\circ e^{\hbar D_{\rm{AF}}}(F\otimes G)\,,\quad D_{\rm{AF}}(F\otimes G)= \left\langle \Delta^{\rm{AF}}, \frac {\delta F }{\delta \varphi} \otimes \frac {\delta G}{\delta \varphi} \right\rangle\,.
\]
Hence
\[
\alpha_H\circ \cS(\lambda\no{V})^{\star-1}=e^{-i\lambda V/\hbar}_{\overline{\cT_H}}\,.
\]

\medskip

We now apply Bogolubov's formula to the case where  $F=\partial_\mu\Phi(f)$ (a component of the current) i.e. $\partial_\mu\Phi(f)(\ph)=\int \partial_\mu\ph(x)f(x)dx$.
In order to apply \eqref{eq:Rn}, first we need to compute the $l+1$-fold time-ordered products of $V^{\otimes l}$ with $\partial_\mu\Phi(f)$. Using formula \eqref{eq:TnH_comb} we obtain
\begin{multline}\label{T:of:del:Phi}
\alpha_H\circ\cT_{l+1}({\no{V}}^{\otimes l}\otimes \partial_\mu\Phi(f))=\mu\circ e^{\hbar \sum_{1\leq i<j\leq l}D_{\mathrm{F}}^{ij}}\circ e^{\hbar \sum_{i=1}^lD_{\mathrm{F}}^{i\, l+1}} (V^{\otimes l}\otimes \partial_\mu\Phi(f))\\
=\cT^H_l(V^{\otimes l})\cdot  \partial_\mu\Phi(f)+l\hbar\, \cT^H_l\left(V^{\otimes l-1}\otimes \int \tsf{\delta V}{\delta \ph(x)}\partial_\mu\Delta^{\mathrm{F}}(x,y)f(y)\right)
\end{multline}
where $D_{\mathrm{F}}^{ij}\doteq \langle\Delta^{\mathrm{F}},\frac{\delta^2}{\delta\ph_i\delta\ph_j}\rangle$. It follows that \eqref{eq:Rn} can be written as a sum of two terms.
\[
R_n({\no{V}}^{\otimes n},F)=J_n({\no{V}}^{\otimes n},F)+M_n({\no{V}}^{\otimes n},F)\,,
\]
so $F_I=J({\no{V}},F)+M({\no{V}},F)$. $J({\no{V}},F)$ is easily computed using the following lemma.
\begin{lem}
	Let $K$ be an integral kernel and let $A,B,C$ be smooth functionals, with $C$ linear, then
	\[
	\mu \circ e^{\hbar D_K}(A\otimes B\cdot C)=\mu \circ e^{\hbar D_K}(A\otimes B)\cdot C+ \mu\circ e^{\hbar D_K}\left(\left<A^{(1)},K C^{(1)}\right>\otimes B\right)\,,
	\]
	where $A^{(1)}$, $C^{(1)}$ are functional derivatives.
\end{lem}
\begin{proof}
The result follows directly from the Leibniz rule. 
\end{proof}
The first contribution to $J({\no{V}},F)$ is just $\partial_\mu\Phi(f)$, since in our case $A$ is the $\star_H$ inverse of $B$. The second contribution can be expanded into time-ordered and antichronological products and using the explicit expression for the vertex operators we obtain a sum of terms of the form
\[
i\hbar \sum_{j=1}^l a_j\bar{\cT}_l^H\left(V_{a_1}(f_1)\otimes\dots\otimes  \int V_{a_j}(x_j) f_j(x_j)\partial_\mu W(x_j,y)f(y) \otimes \dots\otimes V_{a_l}(f_l)\right)\,,
\]
which are then $\star$-multiplied from the right with appropriate time-ordered products of vertex operators.

As for $M({\no{V}},F)$, we also use the explicit form of the vertex operators and obtain
\[
i\hbar \sum_{j=1}^l a_j\cT_l^H\left(V_{a_1}(f_1)\otimes\dots\otimes  \int V_{a_j}(x_j) f_j(x_j)\partial_\mu\Delta^{\mathrm{F}}(x_j,y)f(y) \otimes \dots\otimes V_{a_l}(f_l)\right)\,,
\]
again, $\star$-multiplied from the left with anti chronological products.

We see, in particular, that each contribution to $R_n$ in formula \eqref{eq:Rn} is proportional to products of $n$ vertex operators. 

We now approach the issue of convergence for the interacting fields again from two perspectives: the weak limit in a Hadamard state and the $m\rightarrow 0$ limit of massive vacuum expectation value.

We start with the latter perspective and introduce an auxiliary finite mass $m$, replace all the propagators in the formula for the interacting field by $\Delta^{\mathrm{F}}_m$ and $W_m$ respectively, and change the normal-ordering prescription by applying $\alpha_h^{-1}$ to the vertex operators. The resulting formulea differ from those for the $S$-matrix only by the occurence of factors $\partial_\mu\Delta^{\mathrm{F}}_m$ and $\partial_\mu W_m$. These, however, have well-defined limits for $m\rightarrow 0$. Therefore, we can conclude, as in section \ref{sec:S_matrix}, that only the terms with $\sum_{i=1}^n a_i=0$ contribute in the limit.  Observe that this line of argument is possible for the current $\partial_\mu\Phi_x$, but would not work for $\Phi_x$ itself as $\Delta^{\mathrm{F}}_m$ and $W_m$ are not well-defined in the limit $m \ria 0$. This is again a consequence of the well known infrared problem of massless scalar fields in 2 dimensional Minkowski spacetime.

 We can now proceed as in Theorem~\ref{prop:main}, taking into account the extra factors of $\partial_\mu\Delta^{\mathrm{F}}*f$ and $\partial_\mu W*f$. By a standard result  (see e.g. \cite[\S 9.2]{Folland}), these convolutions are actually smooth functions. Since they are then multiplied with a test function, the result is again a test function.

Finally, note that $\Delta^{\rm AF}$, $\Delta^{\rm F}$, $W$ only differ in the signs in front of the Heaviside functions. The latter did not enter in the estimate in Theorem~\ref{prop:main}, so we conclude that the same estimate as in Theorem~\ref{prop:main} holds for \eqref{eq:Rn}. To conclude, we have shown that the formal power series $\lim_{m\rightarrow 0}\omega_{\ph,H_m}(\partial_\mu\Phi(f)_I)$ is summable.

 Now we briefly discuss the other perspective which is to compute the expectation values of interacting fields in a Hadamard state. Again,  the estimates from Theorem~\ref{prop:main} to conclude that $\omega_{\ph,H_v}(\partial_\mu\Phi(f)_I)$ is summable if we choose $f$ and the cutoff  forthe interaction $V$ appropriately. Contrary to the situation with the auxiliary mass and the limit $m\rightarrow 0$, this is in fact, true even for the interacting field,  $\omega_{\ph,H_v}(\Phi(f)_I)$, since the Hadamard state is better behaved than the singular vacuum state. Our hope is that one can use these results to construct the net of interacting fields, as outlined, for example in section 4.3 of \cite{Vienna}, but this would require a careful choice of test functions and taking the algebraic adiabatic limit.

\section{Conclusion and outlook}
We have proven the convergence of the vacuum expectation value of the Epstein Glaser S-matrix and the interacting current in the finite regime of the Sine-Gordon model. To our knowledge, this is the first such proof done directly in the Lorentzian signature.

The main input into the proof of the main estimate (Theorem~\ref{prop:main})  was the fact that the retarded propagator  
$\Delta^{\rm R}$ is idempotent and bounded, and that the Hadamard function $H$ is a logarithm so that exponentials of $H$ simply give rational functions. Just like  in the earlier investigations~\cite{Froe76,DimockHurd}, this ties the method of proof to 2-dimensional theories.

We expect that it should be possible to apply the above arguments directly to 2-dimensional curved spacetimes where the 2-point function is known explicitly.
A slightly more ambitious goal is to extend our analysis to the superrenormalizable case with $1\leq \beta <2$, extending the ideas of \cite{DimockHurd}.

Most importantly, we believe that the net of interacting fields should be constructable. Our result should open up the possibility  to clarify a number of questions, such as the direct comparison with the (Wick rotated) Euclidean theory, and perhaps even on the factorizability of the $S$-matrix.

\section{Acknowledgment}
We thank David B\"ucher, Harald Grosse and John Imbrie for fruitful discussions in G\"ottingen and Oberwolfach. We also thank Klaus Fredenhagen for useful comments on an earlier version of the manuscript and a suggestion how to prove Proposition 4 in its present, more general form.

\begin{appendix}\section{Calculations}\label{app:calc}

\noindent\textbf{Proof of \eqref{eq:Hexplicit}}

\medskip
	We use the following identities
	\[
	\lim_{\eps \searrow 0} \ln(x+i\eps)=\ln(|x|) - i \pi \theta(x)
	\qquad 
	\lim_{\eps \searrow 0} \ln(x-i\eps)=\ln(|x|) + i \pi \theta(x)
	\]
	and (using $x^2=t^2-\bs{x}^2$)
	\beqa
	-\tsf  1 {4\pi} \ln(-x ^2 + i \varepsilon t)
	&=& 
	-\tsf 1 {4\pi} \ln(\underbrace{-t^2+ \boldsymbol x^2+ i \varepsilon t}_{(\boldsymbol x + (t-i\varepsilon))(\boldsymbol x - (t-i\varepsilon))}) \ = \
	-\tsf 1 {4\pi} \big( \ln(\boldsymbol x + (t-i\varepsilon))+ \ln(\boldsymbol x - (t-i\varepsilon))\big)
	\\
	&=&-\tsf 1 {4\pi} \big( \ln(\boldsymbol x + t-i\varepsilon)+ \ln(\boldsymbol x - t+i\varepsilon)\big)
\eeqa
which in the limit $\varepsilon \searrow 0$ yields
\beqa
&\ria&	-\tsf 1 {4\pi} \big( \ln(|\boldsymbol x + t|)+i\pi \theta (\boldsymbol x + t) + 
	\ln(|\boldsymbol x - t|)- i \pi \theta (\boldsymbol x - t) \big)
	\\&=&
	-\tsf 1 {4\pi} \big( \ln(|\boldsymbol x + t|) + 
	\ln(|\boldsymbol x - t|) \big)
	-\tsf i {4} \big(  \theta (\boldsymbol x + t) - \theta (\boldsymbol x - t) \big)
	\\
	&=&-\tsf 1 {4\pi} \big( \phantom{\ln(|\boldsymbol x + t|) + 
		\ln(|\boldsymbol x - t|)} \big)
	-\tsf i {4} \big(  \theta (t-|\boldsymbol x|) - \theta (-t-|\boldsymbol x|) \big)
	\eeqa
In the last step, we have used the fact that
	\[
	\theta (\boldsymbol x + t) - \theta (\boldsymbol x - t) =
	\left\{\begin{array}{ll}1 &t\geq -\boldsymbol x \text{ and }t\geq \boldsymbol x \\
	-1 &t\leq \boldsymbol x \text{ and }t\leq - \boldsymbol x\\ 
	0 & \text{otherwise}\end{array}\right.
	=
	\left\{\begin{array}{ll}1 &t\geq |\boldsymbol x| \\
	-1 &t\leq -|\boldsymbol x| 
	\\ 
	0 & \text{otherwise}\end{array}\right.
	= \theta (t-|\boldsymbol x|) - \theta (-t-|\boldsymbol x|) 
	\]
A similar computation shows that
	\[
	\theta (t-|\boldsymbol x|) - \theta (-t-|\boldsymbol x|) =\theta(t^2-{\bs{x}}^2)\ \sgn t
	\]
and
	\[
	\theta (t-|\boldsymbol x|) + \theta (-t-|\boldsymbol x|) =\theta(t^2-{\bs{x}}^2)\,.
	\]
Hence the 2-point function can be expressed as
	\[
	W(x)=-\tsf  1 {4\pi} \ln(-x^2+i\varepsilon t)=
	-\tsf  1 {4\pi}\ln(|x ^2|)+i\pi \theta(t^2-\bs{x}^2)\ \sgn t\,.
	\]
A similar reasoning shows that 
	\beqa
	-\tsf  1 {4\pi} \ln(x^2 - i \epsilon)&=& -\tsf  1 {4\pi} \ln(|x^2|)-\tsf i 4 \theta(t^2-\bs{x}^2)\\
	&=& -\tsf  1 {4\pi} \ln(|x^2|)-+\tsf i 2(	-\tsf 1 2 \theta (t-|\boldsymbol x|) -\tsf 1 2 \theta (-t-|\boldsymbol x|))\\
	&=&H+\tsf i 2 (\Delta^{\mathrm{R}}+\Delta^{\mathrm{A}})=\Delta^{\mathrm{F}}\,.
	\eeqa	

\medskip
\noindent\textbf{Proof of \eqref{eq:DeltaDn}}

	First note that $\Delta(x)=\Delta^R-\Delta^A=-\frac 1 2 (\theta(t-|\bs{x}|)-\theta(-t-|\bs{x}|))=-\frac 1 2 \sgn(t)\theta(t^2-\bs{x}^2)$, and since $\Delta$ is the antisymmetric part of the 2-point function, we deduce
	\beqa
	\theta(t) \Delta(x) + \theta(-t)\Delta(-x)
	&=&
	(\theta(t)-\theta(-t)) \Delta(x) \ = \ \sgn(t) \left( - \tsf 1 2 \right)
	\sgn(t)\theta(t^2-\bs{x}^2)
	\\ 
	&=& - \tsf 1 2\theta(t^2-x\bs{x}^2)= -\tsf 1 2 (\theta(t -|\bs{x}|) + \theta(-t-|\bs{x}|)) = 2 \Delta_D
	\eeqa
	Moreover, we have for $n\geq 0$, 
	\beqan\label{eq:Diracn}
	\theta(t) \Delta(x)^n + \theta(-t)\Delta(-x)^n
	&=&
	\left(-\tsf 1 2\right)^n \left(\theta(t) +(-1)^n \theta(-t) \right)\sgn(t)^n\theta(t^2-\bs{x}^2)^n 
	\nonumber \\
	& =& \left\{\begin{array}{ll}
		\left(-\tsf 1 2\right)^n \theta(t^2-\bs{x}^2) \phantom{\int\limits_x}& n \text{ even}\\
		\left(-\tsf 1 2\right)^n \sgn(t)^{n+1}
		\theta(t^2-\bs{x}^2)& n \text{ odd}  \phantom{\int\limits_x}
	\end{array} \right.
	\nonumber 
	\\ 
	&=& 
	(2 \Delta_D(x))^n
	\eeqan
	where in the second equality, for $n$ odd, we used $(\theta(t) +(-1)^n \theta(-t))\sgn(t)^n=\sgn(t)^{n+1}$, and in the last line, we used $- \tsf 1 2\theta(t^2-\bs{x}^2)= 2 \Delta_D$ and $\theta=\theta^n$.

\end{appendix}
{\setlength{\bibsep}{0.5ex}
\bibliographystyle{amsalpha}
\bibliography{References}}
\end{document}